\documentclass[11pt]{article}

\usepackage{fullpage,amsmath,amsthm}
\usepackage{amssymb,latexsym,graphicx}
\usepackage{hyperref}

  \usepackage[T1]{fontenc}
  \usepackage{tgpagella}
  \usepackage{mathpazo}
  \usepackage[heavycircles]{stmaryrd}
\newcommand{\setft}[1]{\mathrm{#1}}
\newcommand{\lin}[1]{\setft{L}\left(#1\right)}

\newcommand{\C}{\ensuremath{\mathbb{C}}}
\newcommand{\N}{\ensuremath{\mathbb{N}}}
\newcommand{\R}{\ensuremath{\mathbb{R}}}
\newcommand{\pmset}[1]{\{-1,1\}^{#1}}

\newcommand{\pauli}[1]{\mathcal{P}_{#1}} 
\newcommand{\sphere}[1]{\setft{S}(\C^{#1})} 
\newcommand{\mat}[1]{\setft{Mat}(#1)} 
\newcommand{\hball}[1]{\setft{B}(\herm{#1})} 
\newcommand{\Ht}{\mathcal{H}}

\newcommand{\herm}[1]{\setft{Herm}\left(#1\right)}
\newcommand{\obs}[1]{\setft{Obs}\left(#1\right)}

\newcommand{\proj}[1]{\setft{Proj}\left(#1\right)}
\newcommand{\projnorm}[1]{\overline{\setft{Proj}}\left(#1\right)}
\newcommand{\net}{\mathcal Z}

\newcommand{\ket}[1]{|#1\rangle}
\newcommand{\bra}[1]{\langle#1|}
\newcommand{\ketbra}[2]{|#1\rangle\!\langle#2|}
\newcommand{\braket}[2]{\langle #1|#2\rangle}
\newcommand{\psiket}{\ket{\psi}}
\newcommand{\psibra}{\bra{\psi}}
\newcommand{\phiket}{\ket{\phi}}
\newcommand{\phibra}{\bra{\phi}}

\newcommand{\normal}{\mathcal{N}\!(0,1)} 
\newcommand{\ie}{i.e.} 
\newcommand{\st}{:\:} 

\newcommand{\Tr}{\mbox{\rm Tr}} 

\newcommand{\Exp}{{\mathbb{E}}}

\DeclareMathOperator{\sign}{sign}

\newcommand{\rank}{\mbox{\rm rank}}
\renewcommand{\Pr}{\mbox{\rm Pr}}

\newcommand{\beq}{\begin{equation}}
\newcommand{\eeq}{\end{equation}}
\newcommand{\beqn}{\begin{equation*}}
\newcommand{\eeqn}{\end{equation*}}
\newcommand{\beqr}{\begin{eqnarray}}
\newcommand{\eeqr}{\end{eqnarray}}
\newcommand{\beqrn}{\begin{eqnarray*}}
\newcommand{\eeqrn}{\end{eqnarray*}}
\newcommand{\be}{\begin{eqnarray}}
\newcommand{\ee}{\end{eqnarray}}

\newcommand{\eps}{\varepsilon}

\newtheorem{theorem}{Theorem}
\newtheorem{proposition}[theorem]{Proposition}
\newtheorem{lemma}[theorem]{Lemma}
\newtheorem{claim}[theorem]{Claim}
\newtheorem{fact}[theorem]{Fact}
\newtheorem*{theorem3}{Theorem~3}
\newtheorem*{theorem2}{Theorem~2}
\newtheorem{corollary}[theorem]{Corollary}
\newtheorem{definition}[theorem]{Definition}
\newtheorem*{remark}{Remark}

\begin{document}

\title{Explicit lower and upper bounds on the entangled value of multiplayer XOR games}
\author{Jop Bri\"{e}t\footnote{CWI. Supported by an NWO Vici grant and the EU project QCS. Email: j.briet@cwi.nl. } \and Thomas Vidick\footnote{Computer Science and Artificial Intelligence Laboratory, Massachusetts Institute of Technology. Supported by the National Science Foundation under Grant No. 0844626. Email: vidick@csail.mit.edu}}

\date{}
\maketitle

\begin{abstract}
XOR games are the simplest model in which the nonlocal properties of entanglement manifest themselves. When there are two players, it is well known that the bias --- the maximum advantage over random play --- of entangled players can be at most a constant times greater than that of classical players. Recently, P\'{e}rez-Garc\'{i}a et al. [Comm. Math. Phys. 279 (2), 2008] showed that no such bound holds when there are three or more players: the advantage of entangled players over classical players can become unbounded, and scale with the number of questions in the game. Their proof relies on non-trivial results from operator space theory, and gives a non-explicit existence proof, leading to a game with a very large number of questions and only a loose control over the local dimension of the players' shared entanglement.

We give a new, simple and explicit (though still probabilistic) construction of a family of three-player XOR games which achieve a large quantum-classical gap (QC-gap). This QC-gap is exponentially larger than the one given by P\'{e}rez-Garc\'{i}a et. al. in terms of the size of the game, achieving a QC-gap of order $\sqrt{N}$ with $N^2$ questions per player. In terms of the dimension of the entangled state required, we achieve the same (optimal) QC-gap of $\sqrt{N}$ for a state of local dimension $N$ per  player.  Moreover, the optimal entangled strategy is very simple, involving observables defined by tensor products of the Pauli matrices. 

Additionally, we give the first upper bound on the maximal QC-gap in terms of the number of questions per player, showing that our construction is only quadratically off in that respect. 
Our results rely on probabilistic estimates on the norm of random matrices and higher-order tensors which may be of independent interest.

\end{abstract}

\section{Introduction}
Multiplayer games, already a very successful abstraction in theoretical computer science, were first proposed as an ideal framework in which to study the nonlocal properties of entanglement by Cleve et al.~\cite{Cleve:04a}. Known as \emph{nonlocal}, or \emph{entangled}, games, they can be thought of as an interactive re-framing of the familiar setting of Bell inequalities: a referee (the experimentalist) interacts with a number of players (the devices). The referee first sends a classical question (a setting) to each player. The players are all-powerful (there is no restriction on the shared state or the measurements applied) but not allowed to communicate: each of them must make a local measurement on his or her part of a shared entangled state, and provide a classical answer (the outcome) to the referee's question. The referee then decides whether to accept or reject the players' answers (he evaluates the Bell functional). 

In their paper, Cleve et al.\ gave an in-depth study of the simplest class of multiplayer games, \emph{two-player XOR games}. The XOR property refers to the fact that in such games each player answers with a single bit, and the referee's acceptance criterion only depends on the parity of the bits he receives as answers. One of the most fundamental Bell inequalities, the CHSH inequality~\cite{ClauserHSH:1969}, fits in this framework. In the corresponding XOR game the acceptance criterion dictates that the parity of the players' answers must equal the product of their questions, a uniform i.i.d.\ bit each. The laws of quantum mechanics predict that the CHSH game has the following striking property: there is a quantum strategy in which the players share a simple entangled state --- a single EPR pair --- and use it to achieve a strictly higher success probability than the best classical, unentangled strategy: roughly $85\%$, as compared to $75\%$. This example demonstrates that quantum mechanics is \emph{nonlocal}: predictions made by the theory cannot be reproduced classically, or more generally by any local hidden variable model, a ``paradox'' most famously put forward by Einstein, Podolsky and Rosen~\cite{Einstein:35a}.

\medskip

Any XOR game $G$ can be won with probability~$1/2$ by players who independently answer each question with the outcome of a random coin flip. It is therefore natural to measure the success of quantum (resp. classical) players through their maximum achievable \emph{bias} $\beta^*(G)$ (resp.~$\beta(G)$), defined as their maximum winning probability in the game, \emph{minus} the success probability that would be achieved by random play. As has become standard practice, we will measure the advantage of quantum over classical players through the ratio $\beta^*(G)/\beta(G)$, referred to as the quantum-classical gap, or {\em QC-gap} for short.\footnote{See Section~\ref{sec:conclusion} for a brief discussion of other ways of measuring the quantum advantage, such as through the difference $\beta^*(G)-\beta(G)$.}
The CHSH example demonstrates the existence of a game for which $\beta^*(G) \geq \sqrt{2} \beta(G)$, and Tsirelson~\cite{Tsirelson:85b} proved that this gap was close to best possible. By making a connection to the celebrated {\em Grothendieck inequality} he showed that for any two-player XOR game $G$, we have $\beta^*(G)/\beta(G) \leq K_G^\R$, where $K_G^\R$ is the real Grothendieck constant.\footnote{The subscript $G$ in $K_G^\R$ stands for ``Grothendieck'', and is not related to the game $G$!} The exact value of $K_G^\R$ is unknown, and the best upper bound currently known, $K_G^\R\lesssim 1.78$, appeared in recent work of Braverman et. al.~\cite{BMMN11}. 
Although experiments based on the CHSH game have been performed~\cite{Aspect:1981,Aspect:1982}, the relatively small gap forces the use of state-of-the-art devices in terms of precision and timing in order to differentiate a truly nonlocal strategy from one that can be explained by local hidden-variable models. 
In order to observe larger quantum-classical gaps, more general classes of games need to be considered, prompting a question that has driven much recent research in this area:
{\em For a given QC-gap, what is the simplest game (in terms of the number of players, questions and answers) which demonstrates such a QC-gap (if one at all exists)?} 

\medskip 

There are two main directions in which one can look for generalizations of two-player XOR games. The first is to increase the number of possible answers from each player. This option has so far been the preferred one, and has by now been relatively well explored~\cite{Cleve:04a,kempe:unique,JPPVW10,JP11,Reg11,BRSW11}. In particular it is known that the largest possible quantum-classical gap is bounded by a constant times the minimum of the number of questions, the number of answers, and the local dimension of the players~\cite{JPPVW10}, and there are explicit constructions of games (\ie, games games whose existence is proved through a constructive proof) which come close to achieving these bounds~\cite{BRSW11}.
Unfortunately, these games require the players to perform complex measurements, involving large numbers of outcomes, making them ill-suited to experiment.  

The second possible avenue for generalization consists in increasing the number of players, while remaining in the simple setting of binary answers and an XOR-based acceptance criterion. 
Our limited understanding of multipartite entanglement makes this setting more challenging, and for a long time little more than small, constant-size examples were known~\cite{Mermin:1990,Zukowski:1993a}.  However, recently, P\'{e}rez-Garc\'{i}a et al.~\cite{perezgarcia:2008} discovered that adding even just one player allowed for a very different scaling of the QC-gap. They demonstrated the existence of an infinite family of three-player XOR games $(G_N)_{N\in\N}$ for which $\lim_{N\to\infty} \beta^*(G_N)/\beta(G_N) = +\infty$ --- an unbounded gap! This exciting result demonstrated for the first time that very large violations could be observed even in the relatively simple context of three-player XOR games. 

\medskip

The results in~\cite{perezgarcia:2008} were proved by establishing a surprising connection between XOR games and certain natural norms on the tensor product of operator spaces, enabling the authors to leverage powerful techniques from the latter area in order to establish their results on XOR games. Since their seminal paper, similar techniques have been successfully applied to other settings, such as general two-player games~\cite{JPPVW10} and games with quantum communication~\cite{CJPP11}.

For the games $G_N$ from~\cite{perezgarcia:2008}, however, the above-mentioned techniques have a few somewhat unfortunate consequences.
First of all, these techniques resulted in a highly non-explicit existence proof. While P\'{e}rez-Garc\'{i}a et al.\ show the \emph{existence} of the games, it seems quite hard to even get the slightest idea of what the games would look like.  
Moreover, their use of the theory of operator spaces gives a very large game, with an exponential (in the QC-gap) number of questions per player. 
Finally, the strategies required of the players to achieve the promised QC-gap are not explicitly known, and may for instance require an entangled a state with unbounded dimension on two of the players; only the first player's dimension is controlled. We note that after the completion of our work, but independently from it, Pisier~\cite{PisierGT,Pisiernote} showed that the construction in~\cite{perezgarcia:2008} could be improved to require only a polynomial number of questions to each player, and that one could keep a control of the entanglement dimension on all three players. The resulting parameters, however, are still worse than the ones that we achieve here. 

\subsection{Our results} 

In this paper we give a new and improved proof of the existence of a family of three-player XOR games for which the QC-gap is unbounded. Our proof technique uses the probabilistic method: we describe a simple \emph{probabilistic} procedure that outputs a game with the desired properties \emph{with high probability}. As such it is much more explicit than previous results~\cite{perezgarcia:2008}, albeit not fully constructive. 
Our construction is outlined in Section~\ref{sec:techniques} below. For a desired ratio $\sqrt{N}$, our game has order $N^2$ questions per player, which, as we show, is within a factor $\tilde{O}(N)$ of the smallest number possible. Moreover, to achieve such a gap entangled players only need to use Pauli observables and  an entangled state of local dimension $N$ per player. 
The simplicity of our construction enables us to give concrete values for most of the parameters, leading to a rigorous control of the constants involved. We prove the following: 

\begin{theorem}\label{thm:main} For any integer $n$ and $N=2^n$ there exists a three-player XOR game $G_N$, with $N^2$ questions per player, such that $\beta^*(G_N) \geq \Omega(\sqrt{N}\log^{-5/2} N)\, \beta(G_N)$. Moreover, there is an entangled strategy which achieves a bias of $\Omega(\sqrt{N}\log^{-5/2} N)\,\beta(G_N)$, uses an entangled state of local dimension $N$ per player, and in which the players' observables are tensor products of $n$ Pauli matrices.
\end{theorem}

Additionally, we prove that the dependence of the QC-gap on the number of questions obtained in Theorem~\ref{thm:main} is close to optimal.\footnote{A similar result was recently communicated to us by Carlos Palazuelos~\cite{Carlos:personal}.} This improves upon an independent previous result by Loubenets~\cite{Lou12}, who showed that $\beta^*(G) \,\leq\, (2Q-1)^{2}\,\beta(G)$.

\begin{theorem}\label{thm:questionbound}
For any $3$-player XOR game $G$ in which there are at most $Q$ possible questions to the third player,
 $$\beta^*(G) \,\leq\, \sqrt{Q}\,K_G^\R\,\beta(G),$$
 where $K_G^\R<1.783$ is the real Grothendieck constant.
\end{theorem}

Finally, we also show that the dependence on the local dimension of the entangled state is optimal, re-proving in a simpler language a result first proved in~\cite{perezgarcia:2008}.

\begin{theorem}\label{thm:dimensionbound}
Let $G$ be a $3$-player XOR game in which the maximal entangled bias $\beta^*(G)$ is achieved by a strategy in which the third player's local dimension is $d$. Then 
$$\beta^*(G) \,\leq\, \sqrt{3d}\, \big(K_G^\C\big)^{3/2} \beta(G),$$
where $K_G^\C < 1.405$ is the complex Grothendieck constant.
\end{theorem}

\paragraph{Generalizations.} While we present our results in the case of three-player XOR games, they have straightforward extensions to an arbitrary number of players. In particular, one can show that the following holds, for any $r\geq 3$:
\begin{enumerate}
\item For any integer $N$ that is a power of $2$, there exists a $r$-player XOR game $G$, with $N^2$ questions per player, such that $\beta^*(G) \geq \Omega\big((N\log^{-5} N)^{(r-2)/2}\big)\beta(G)$, and there is a entangled strategy achieving this gap that involves only $N$-dimensional Pauli observables. 
\item If $G$ is a $r$-player XOR game in which at least $r-2$ of the players have at most $Q$ possible questions each, then $\beta^*(G) \leq O(Q^{(r-2)/2})\beta(G)$.
\item If $G$ is a $r$-player XOR game in which the shared state of the players is restricted to have local dimension $d$ on at least $r-2$ of the players, then $\beta^*(G) \leq O\big(d^{(r-2)/2}\big)\beta(G)$.
\end{enumerate}

\paragraph{Applications to operator space theory.} The original motivation for the construction in~\cite{perezgarcia:2008} was to show that a certain trilinear extension of Grothendieck's inequality does not hold. Our construction leads to an improved obstruction: the three-player XOR game constructed in Theorem~\ref{thm:main} can be used to show that two different norms are not equivalent for the space of trilinear functionals on~$\ell_\infty\times\ell_\infty\times\ell_\infty$.
More precisely, Theorem~\ref{thm:main} implies that for any $N=2^n$ there is a trilinear functional $T: \ell_\infty^{N^2} \times \ell_\infty^{N^2} \times \ell_\infty^{N^2}\to\C$ such that 
\beq\label{eq:os-norm}
\| T \|_{\text{cb}}\, \geq \, \Omega\big(\sqrt{N}\log^{-5/2} N\big)\,\| T\|,
\eeq
and the $N$th amplifications in the completely bounded norm suffice.
Put differently, the injective and minimal tensor norms are inequivalent on~$\ell_1\otimes\ell_1\otimes\ell_1$.
This improves on the estimate from~\cite{perezgarcia:2008}, in which the bound was logarithmic in $N$.
For more details and background on relevant aspects of Grothendieck's inequality we refer to the excellent survey~\cite{PisierGT}, and to Section 20 in particular for the connection with XOR games.

In addition, Pisier~\cite{Pisiernote} recently applied our result to prove an almost-tight estimate on the norm of the re-ordering map 
\beq\label{eq:defJ}
 J:\, (H_1\otimes_2 K_1) \otimes_\epsilon \cdots \otimes_\epsilon (H_\ell \otimes_2 K_r)\,\to\, (H_1\otimes_2 \cdots \otimes_2 H_r) \otimes_\epsilon (K_1\otimes_2 \cdots \otimes_2 K_r), 
\eeq
where $H_i,K_i$ are $N$-dimensional Hilbert spaces, proving that $\| J \| = \tilde{\Omega}\big(N^{r-1}\big)$, where the $\tilde{\Omega}$ notation ignores possible poly-logarithmic factors. 

\begin{samepage}
\subsection{Proof overview and techniques}\label{sec:techniques}

\paragraph{Lower bound.} Our construction of a three-player XOR game $G_N$ proceeds through two independent steps. 
\end{samepage}
In the first step we assume given a $3$-tensor $T = T_{(i,i'),(j,j'),(k,k')}$ of dimension $N^2\times N^2\times N^2$, where $N$ is a power of $2$. Based on $T$, we define a three-player XOR game $G_N=G(T)$. Questions in this game are $N$-dimensional Pauli matrices $P,Q,R$, and the corresponding game coefficient\footnote{The equation below defines a complex number. Taking its real or imaginary part would result in a Bell functional, which can in turn easily be transformed into an XOR game through a proper normalization.} is defined as
$$ G(P,Q,R) \,=\,\langle T,P\otimes Q \otimes R \rangle \,:=\, \sum_{(i,i'),(j,j'),(k,k')} T_{(i,i'),(j,j'),(k,k')} \, P_{i,i'} Q_{j,j'} R_{k,k'}.$$
This definition results in a game whose entangled and classical biases can be directly related to \emph{spectral properties} of the tensor $T$.  On the one hand we show that the classical bias $\beta(G_N)$ reflects the tripartite structure of $T$, and is upper-bounded by the norm of $T$ as a trilinear operator. On the other hand we show that the entangled bias $\beta^*(G_N)$ is \emph{lower-bounded} by the norm of $T$ \emph{as a matrix} --- a bilinear operator on $N^3$-dimensional vectors, obtained by pairing up the indices $(i,j,k)$ and $(i',j',k')$. This new connection reduces the problem of constructing a game with large QC-gap to constructing a tensor $T$ with appropriate spectral properties. 

The second step of the proof is our main technical contribution. We give a probabilistic construction of a $3$-tensor $T$ having large norm when seen as a bilinear operator (giving a large entangled bias), but low norm when seen as a trilinear operator (giving a low classical bias). 
To this end, we simply take~$T$ to correspond to an (almost) rank-$1$ matrix: letting $(g_{ijk})$ be a random $N^3$-dimensional vector with i.i.d. entries distributed as standard Gaussians,\footnote{Our results also hold with the Gaussians replaced by i.i.d. standard Bernoulli random variables.} the $(i,i'),(j,j'),(k,k')$-th entry of $T$ is $g_{ijk}\,g_{i'j'k'}$ if $i\neq i'$, $j\neq j'$ and $k\neq k'$, and~$0$ otherwise. 
The fact that $T$, when seen as a matrix, is close to having rank~1 makes it easy to lower bound its spectral norm. An upper bound on the norm of $T$ as a trilinear operator is proved in two steps. In the first step we apply a concentration bound due to Lata\l a to show that for any fixed Hermitian $X,Y,Z$ with Frobenius norm at most~$1$, the product $|\langle T,X\otimes Y \otimes Z\rangle|$ is highly concentrated around its expected value, where the concentration is over the random choice of~$T$. We then conclude by a union bound, using a delicate $\eps$-net construction based on a decomposition of Hermitian matrices with Frobenius norm at most $1$ as linear combinations of (normalized, signed) projectors. 

\paragraph{Upper bounds.} 
We prove upper bounds on the largest possible QC-gap achievable by any three-player XOR game, both as a function of the local dimension of an optimal strategy, and of the number of questions per player in the game. Both bounds follow the same overall proof strategy: using a decoupling argument, we show that the third player can be restricted to applying a classical strategy while incurring only a bounded factor loss in the bias. We conclude by applying (the easy direction of) Tsirelson's Theorem and Grothendieck's inequality (see Section~\ref{sec:grothendieck}) to show that the first two players can be made classical at a further loss of a constant factor only.    


\paragraph{Organization of the paper.} We start with some preliminaries in Section~\ref{sec:prelim}. We describe our construction of a game with unbounded QC-gap in Section~\ref{sec:unbounded}. Our upper bounds on the QC-gap as a function of the number of questions and the local dimension are proved in Section~\ref{sec:ub}. We conclude with some open questions in Section~\ref{sec:conclusion}.

\section{Preliminaries}
\label{sec:prelim}

\subsection{Notation}

For a positive integer $N$ we define $[N]:=\{1,\dots,N\}$. For a positive integer $K$ we denote by~$[N]^K$ the Cartesian product of the set $[N]$ with itself $K$ times (\ie, $[N]\times\cdots\times[N]$).

For a subset $\mathcal W\subseteq \mathcal V$ of a normed vector space $(\mathcal V,\|\cdot\|)$ we let $\setft{S}(\mathcal W) := \{X\in \mathcal W\st \|X\| = 1\}$ be the unit sphere, $\setft B(\mathcal W,\tau) := \{X\ \in\mathcal W\st \|X\| \leq \tau\}$ the ball of radius $\tau$ and  $\setft B(\mathcal W) := \setft B(\mathcal W,1)$ the unit ball. We let $\|\cdot\|_2$ denote the usual Euclidean norm. Throughout we endow $\C^N$ with this norm.

We will usually use $g\sim\normal$ to denote a real-valued random variable distributed according to a standard normal (Gaussian) distribution (\ie, a variable with mean~$0$ and variance~$1$), and $\ket{g}\sim \normal^N$ for an $N$-dimensional vector whose entries are i.i.d. standard normal random variables.

\paragraph{Matrices.} Throughout $\mathcal{H}$ will denote a $N$-dimensional complex Hilbert space. We identify the set of linear operators $\lin{\Ht}$ on $\Ht$ with the set $\mat{N}$ of complex $N$-by-$N$ matrices. Let  $\herm{\Ht} = \{X\in \lin{\Ht}\st X^{\dagger} = X\}$ be the subset of Hermitian operators and  $\obs{\Ht}\subseteq\herm{\Ht}$ be the Hermitian operators with all eigenvalues in $\pmset{}$. In other words, $\obs{\Ht}$ is the set of $\pmset{}$-valued observables on $\mathcal{H}$. 
Note that operators in $\obs{\Ht}$ are unitary and square to the identity.
We will use the notation $\herm{N}$ and $\obs{N}$ when we think of the operators'  matrix representation. The space of matrices $\mat{N}$ is a Hilbert space for the inner product $(A,B)\mapsto \langle A,B\rangle := \Tr(AB^\dagger)$. The resulting norm is the {\em Frobenius norm} $A\mapsto \|A\|_F := \sqrt{\Tr(AA^\dagger)}$.
Throughout we tacitly endow $\mat{N}$ with the Frobenius norm, so the balls and sphere are always defined with respect to this norm.
Note that if we let the singular values of a matrix $X\in \mat{N}$ be $\sigma_1(X)\geq\dots\geq\sigma_N(X)$, then $\|X\|_F^2 = \sigma_1(X)^2 + \cdots + \sigma_N(X)^2$. We recall that for each eigenvalue $\lambda$ of a Hermitian matrix $X$ there is a corresponding singular value $\sigma = |\lambda|$.
We denote by $\|\cdot \|_\infty = \sigma_1(X)$ the operator norm on $\mat{N}$.
Let $\proj{N}_k\subseteq \herm{N}$ be the set of rank-$k$ (orthogonal) projectors on~$\C^N$ and let $\projnorm{N}_k = \proj{N}_k/\sqrt{k}$ be the set of rank-$k$ projectors that are normalized with respect to the Frobenius norm. Define the set of all $N$-dimensional normalized projectors by $\projnorm{N} = \bigcup_{k=1}^N\projnorm{N}_k$. 

If $N = 2^{n}$ for some positive integer $n$, we let $\mathcal{P}_{n} := \big\{\big(\begin{smallmatrix}1&0\\ 0&1\end{smallmatrix}\big), \big(\begin{smallmatrix}0&1\\ 1&0\end{smallmatrix}\big), \big(\begin{smallmatrix}0&-i\\ i&0\end{smallmatrix}\big), \big(\begin{smallmatrix}1&0\\ 0&-1\end{smallmatrix}\big) \big\}^{\otimes n}$ be the set of $n$-fold tensor products of Pauli matrices. The letters $P,Q,R$ will usually denote elements of $\mathcal{P}_n$. We have $|\mathcal{P}_{n}| = N^2$, and for $P,Q\in\mathcal{P}_n$ we have $\langle P,Q\rangle  = N\,\delta_{P,Q}$: the set $\mathcal{P}_n$ forms an orthogonal basis of observables for~$\mat{N}$.

\subsection{Tensors}\label{sec:tensors}

Given positive integers $r,N_1,\dots,N_r$, an {\em $r$-tensor of dimensions $N_1\times\cdots\times N_r$} is a map of the form $T:[N_1]\times\cdots\times[N_r]\to\C$. Every element $T(i_1,\dots,i_r)$ of such a tensor is specified by an $r$-tuple of indices $(i_1,\dots,i_r)\in[N_1]\times\cdots\times[N_r]$. We will mostly deal with $3$-tensors of dimensions $N^2\times N^2\times N^2$ for  some $N\in\N$. In this case we index the elements by three pairs of indices $(i,i'),(j,j')$ and $(k,k')\in[N]^2$. We will think of such a tensor in two different ways: as a \emph{bilinear} functional acting on $N^3$-dimensional complex vectors, and as a \emph{trilinear} functional acting on Hermitian $N\times N$ matrices. For the sake of concreteness we now describe in detail how these two perspectives relate to each other.

\paragraph{The bilinear view.}
Let $T$ be a $3$-tensor of dimension $N^2\times N^2\times N^2$.
The dimensions of the tensor $T$ allow us to view it as an $N^3$-by-$N^3$ complex matrix. Correspondingly, we define the spectral norm of $T$ by
\beqn
\|T\|_{3,3}\,:=\, \max_{x,y \in \sphere{N^3}}  \bigg|\sum_{(i,j,k),(i',j',k')\in[N]^3} T_{(i,i'),(j,j'),(k,k')} x_{i,j,k}y_{i',j',k'}\bigg|.
\eeqn

Suppose that that for some $n\in \N$, we have $N = 2^{n}$.
Since the set $\pauli{N^3} = \{X\otimes Y\otimes Z\st X,Y,Z\in\pauli{N}\}$ is an orthogonal basis for $\mat{N^3}$, we can define the ``Fourier coefficient'' of $T$ at $(P, Q, R)$ as 
\beqn
\widehat{T}(P,Q,R) := \langle T,P\otimes Q\otimes R\rangle = \sum_{(i,i'),(j,j'),(k,k')\in[N]^2} T_{(i,i'),(j,j'),(k,k')} P_{i,i'} Q_{j,j'} R_{k,k'}. 
\eeqn

With this definition, $T$ can be written as
\beqn
T = N^{-3} \sum_{P,Q,R\in\pauli{N}} \widehat{T}(P,Q,R) P\otimes Q \otimes R.
\eeqn

\paragraph{The trilinear view.}
Let $T$ be a $3$-tensor of dimensions $N^2\times N^2\times N^2$.
We can associate with $T$ a trilinear functional $L_T:\herm{N} \times \herm{N} \times \herm{N}\!\to\!\C$ defined by
\beqn
L_T(X,Y,Z) \,=\,  \langle T,X\otimes Y\otimes Z\rangle\,=\, \sum_{(i,j,k),(i',j',k')\in [N]^3} T_{(i,i'),(j,j'),(k,k')} X_{i,i'} Y_{j,j'} Z_{k,k'},
\eeqn
where $X,Y,Z\in\herm{N}$. 
The operator norm of $L_T$ induces the following norm on $T$
\beqn
\|T\|_{2,2,2}\,:= \max_{X, Y, Z\in \hball{N}}  \big|L_T(X,Y,Z)\big| =  \max_{X,Y,Z\in\hball{N}}|\langle T,X\otimes Y\otimes Z\rangle|.\footnote{The restriction to Hermitian matrices in this definition is not essential, but it will be convenient later on.}
\eeqn
%
%

\subsection{XOR games}\label{sec:prelimxor}

An $r$-player XOR game with $N$ questions per player is fully specified by a joint probability distribution $\pi$ on $\mathcal [N]^r$  and an $r$-tensor $M:[N]^r\to \pmset{}$.
The {\em classical bias} of an XOR game $G = (\pi, M)$ is defined by
\beqn
\beta(G) := \max_{\chi_1,\dots,\chi_r:[N]\to\pmset{} }\Exp_{(q_1,\dots,q_r)\sim\pi}\Big[M(q_1,\dots,q_r)\,\chi_1(q_1)\cdots \chi_r(q_r)\Big].
\eeqn
The maps $\chi_1,\dots,\chi_r$ in the above maximum are referred to as {\em strategies}: they should be interpreted as giving the players' answers to the questions $q_1,\dots,q_r$, respectively.
The \emph{entangled bias} of $G$ is defined by
\beqn
\beta^*(G) := \sup_{\substack{d\in\N,\,\ket{\Psi}\in\sphere{d^r}\\ A_1,\dots,A_r:[N]\to\obs{d}}}\Exp_{(q_1,\dots,q_r)\sim\pi}\Big[M(q_1,\dots,q_r)\, \psibra A_1(q_1)\otimes\cdots\otimes A_r(q_r)\psiket\Big].
\eeqn

In the sequel it will be convenient to merge $\pi$ and $M$ into a single tensor $T:[N]^r\to\R$ defined by $T(q_1,\dots,q_r)= \pi(q_1,\dots,q_r)M(q_1,\dots,q_r)$. Conversely, any tensor $T:[N]^r\to\R$ defines (up to normalization) an XOR game by setting the distribution to $\pi(q_1,\dots,q_r) = |T(q_1,\dots,q_r)|$ and the game tensor to $M(q_1,\dots,q_r) = \sign\big(T(q_1,\dots,q_r)\big)$.

\subsection{$\eps$-nets}\label{sec:epsnet}

Our probabilistic proof of the existence of a game for which there is
a large QC-gap relies on the construction of
specific $\eps$-nets over Hermitian matrices, which we describe in
this section.

\begin{definition} An $\eps$-net for a subset $\mathcal W$ of a metric space
$(\mathcal V,d)$ is a finite set $W\subseteq \mathcal V$ such that for every $x\in \mathcal W$,
there exists an $s\in W$ such that $d(x,s)\leq \eps$.
\end{definition}

\begin{fact}\label{fact:epsnet} 
For every $N\in\N$ and any $\eps>0$ there
exists an $\eps$-net $S_\eps$ for $\sphere{N}$ of cardinality $|S_\eps|\leq
(1+2/\eps)^N$.
\end{fact}

\begin{proof} This well-known fact follows from a volume argument: we
can choose $S_\eps$ so that the balls with radius $\eps/2$ centered at
the points in $S_\eps$ are disjoint.
(See e.g.~\cite[Lemma 4.10]{Pisier:1999}.)
\end{proof}

The following lemma shows that for any $\eps>0$ and $N>1$, an $\eps/(4\sqrt{\ln N})$-net over (normalized, signed) $N$-dimensional projections automatically induces an $\eps$-net over $N$-dimensional Hermitian matrices with Frobenius norm at most $1$. The lemma follows from a well known equivalence between the unit ball of normalized projections and the unit ball corresponding to the matrix norm derived from the Lorentz-sequence semi-norm $\ell_{2,1}$. We give a self-contained proof below. 

\begin{lemma}\label{lem:lorentz}
Let $N>1$ and $X\in B\big(\herm{N}\big)$. Then $X$ can be decomposed as a linear combination 
\beqn
 X = \sum_x \lambda_x X_x,
 \eeqn
where each $X_x\in \projnorm{N}$ is a normalized projector and $\sum_x |\lambda_x| \leq 4\sqrt{\ln N}$.  
\end{lemma}


\begin{proof} Let $X = \sum_i \lambda_i \ketbra{u_i}{u_i}$ be the spectral decomposition of a Hermitian matrix $X$ with norm 
\beq\label{eq:lorentz-1}
\big\|X\big\|_F^2 \,=\, \sum_i \lambda_i^2 \,\leq\, 1.
\eeq
For every $t\in [-1,1]$, let $P_t$ be the projector on Span$\{\ket{u_i}:\, \lambda_i \in [-1,-t)\}$ if $t<0$ and the projector on Span$\{\ket{u_i}:\, \lambda_i \in (t,1]\}$ if $t\geq 0$. Then the following holds:
\beq\label{eq:lorentz-2}
 X \,=\, \int_{t=-1}^1 \text{sign}(t)\,P_t\,\text{dt}\,=\, \int_{t=-1}^1 \sqrt{\text{rank} P_t}\,\frac{\text{sign}(t)\,P_t}{\sqrt{\text{rank}P_t}}\,  \text{dt},
\eeq
where the integral is taken coefficient-wise.\footnote{The coefficients of $P_t$ are step functions, so the integral is well-defined.} By a direct calculation,  
\beq\label{eq:lorentz-3}
\int_{t=-1}^{1} |t|\, \Tr(P_t) \text{dt} \,=\, \frac{1}{2}\sum_i \lambda_i^2 \,\leq\, \frac{1}{2},
\eeq
where the last equality follows from~\eqref{eq:lorentz-1}. Eq.~\eqref{eq:lorentz-2} shows that $X$ may be written as a non-negative linear combination of the $\text{sign}(t)\,P_t/\sqrt{\text{rank}P_t}$ with coefficients summing up to
\begin{align*}
\int_{t=-1}^1  \sqrt{\text{rank} P_t}\,  \text{dt} &\leq \int_{t=-1/\sqrt{N}}^{1/\sqrt{N}} \sqrt{N}\,\text{dt} + \int_{t=-1}^{-1/\sqrt{N}} \sqrt{\text{rank} P_t}\,  \text{dt} +\int_{t=1/\sqrt{N}}^1 \sqrt{\text{rank} P_t}\,  \text{dt}\\
&\leq 2+ \Big(2\int_{t=1/\sqrt{N}}^1 \frac{1}{t}\,\text{dt}\Big)^{1/2}\Big(\int_{t=-1}^{-1/\sqrt{N}} (-t)\,\text{rank} P_t\,\text{dt}+\int_{t=1/\sqrt{N}}^1 t\,\text{rank} P_t\,\text{dt}\Big)^{1/2}\\
&\leq 2+ \sqrt{\ln N/2},
\end{align*}
where the first inequality uses $\text{rank} P_t \leq \min\big( N,\, t^{-2}\big)$ for every $t$, the second inequality follows from Cauchy-Schwarz and the last uses~\eqref{eq:lorentz-3}, together with $\text{rank} P_t = \Tr P_t$. 
\end{proof}

The following lemma gives a straightforward construction of an $\eps$-net for the set of normalized rank-$k$ projectors on~$\C^d$ (see e.g.~\cite{Szarek:1982} for more general constructions of nets on Grassmannian spaces).

\begin{lemma}\label{lem:projnet}
For every $k\in[N]$ and any $0<\eps\leq 1$ there exists a set $$\net_\eps^k\subseteq \bigcup_{\ell=1}^k\projnorm{N}_\ell$$ of size $|\net_\eps^k|\leq 2(5/\eps)^{kN}$, such that for any $X\in \projnorm{N}_k$ there is an $\widetilde X\in\net_\eps^k$ satisfying $\|X - \widetilde X\|_F \leq \eps$.
\end{lemma}

\begin{proof}
Let $\eta = \eps/\sqrt{2}$ and $S_{\eta}$ be an $\eta$-net for the unit sphere $\sphere{N}$.
For every $k$-subset $T\subseteq S_{\eta}$ let $Y_T$ be the projector on the space spanned by the vectors in $k$ and let $\overline Y_T = Y_T/\sqrt{\rank Y}$. Define the set $\net_\eps^k$ by
\beqn
\net_\eps^k = \left\{\overline Y_T: T\subseteq S_\eta,\, |T| = k\right\}.
\eeqn
Note that for any $\overline Y\in \net_\eps^k$, we have $\rank (\overline Y) \leq k$. 
Moreover, by the upper bound on the minimal size of~$S_\eta$ from Fact~\ref{fact:epsnet}, we have
\beqn
|\net_\eps^k| \leq {|S_\eta|\choose k} \leq {(3/\eta)^N\choose k} \leq 2\left(\frac{5}{\eps}\right)^{kN}.
\eeqn

Fix $X\in \projnorm{N}_k$ and let $\ket{\phi_1},\dots,\ket{\phi_k}\in\sphere{N}$ be orthonormal eigenvectors of $X$ with eigenvalue~$1/\sqrt{k}$. Let $\ket{\psi_1},\dots,\ket{\psi_k}\in S_\eta$ be the vectors closest to $\ket{\phi_1},\dots,\ket{\phi_k}$ (resp.) with respect to the Euclidean distance. 
Let $Y$ be the projector on the space spanned by the $\ket{\psi_1},\dots,\ket{\psi_k}$ and let $\overline Y= Y/\sqrt{\rank Y}$. Clearly, $\overline Y\in \net_\eps^k$.
Since $Y$ is positive semidefinite and for every $i = 1,\dots, k$, the vector $\ket{\psi_i}$ is an eigenvector of $Y$ with eigenvalue~1, 
\beqn
\bra{\phi_i}Y\ket{\phi_i} \geq |\braket{\phi_i}{\psi_i}|^2 \geq 1 - \eta^2,
\eeqn
where the second inequality follows since $\ket{\psi_i}$ is closest to $\ket{\phi_i}$ in the $\eps$-net.\footnote{Notice that for any complex unit vectors $x,y$, we have $\|x-y\|^2 = 2- 2\Re(\langle x,y\rangle)$ and $|\langle x,y\rangle|^2 = \Re(\langle x,y\rangle)^2 + \Im(\langle x,y\rangle)^2$.}
By definition of the Frobenius norm and the fact that $X$ and $\overline Y$ are Hermitian, we get
\beqrn
\|X - \overline Y\|_F^2 &=& \|X\|_F^2 + \|\overline Y\|_F^2 - 2 \Tr(X\overline Y)\\
&\leq& 2 - 2\Tr(X\overline Y)\\
&\leq& 2\left(1 - \frac{1}{k}\sum_{i=1}^k \bra{\phi_i}Y\ket{\phi_i}\right)\\
&\leq& 2\left(1 - \Big(1-\eta^2\Big) \right)\\
&=& \eps^2,
\eeqrn
and the lemma is proved.
\end{proof}

\begin{definition}\label{def:triplenet}
For every triple of integers $(k,\ell,m)\in[N]^3$ and any real number $0<\eps\leq 1$, define
\beqn
\net_{\eps}^{(k,\ell, m)} = \{X\otimes Y\otimes Z\st (X,Y,Z)\in \net_\eps^k\times \net_\eps^\ell\times \net_\eps^m \},
\eeqn 
and $\net_{\eps} = \bigcup_{(k,\ell,m)\in[N]^3} \net_{\eps}^{(k,\ell, m)}$.
\end{definition}

\begin{proposition}\label{prop:triplenet}
For any $\eps>0$ and $d>1$, the set $\net_{\eps} $ is a $3\eps$-net for the set of matrices $X\otimes Y\otimes Z$ where $(X,Y,Z)\in\projnorm{N}\times\projnorm{N}\times\projnorm{N}$, with respect to the distance function defined by the Frobenius norm. 
\end{proposition}

\begin{proof} Let $1\leq k,\ell,m\leq N$ and $X\in\projnorm{N}_k$, $Y\in\projnorm{N}_\ell$ and $Z\in\projnorm{N}_m$.
Let $\widetilde X\in\net_\eps^k$, $\widetilde Y\in\net_\eps^\ell$ and $\widetilde Z\in\net_\eps^m$ be the closest elements in the nets to (resp.) $X$, $Y$ and $Z$ in Frobenius distance.
Using the trivial identity $A\otimes B - \tilde A\otimes \tilde B = A\otimes (B - \tilde B) - (A - \tilde A)\otimes \tilde B$ twice in a row, and the triangle inequality, we can upper bound the distance~$\| X\otimes Y\otimes Z- \tilde X\otimes \tilde Y\otimes \tilde Z\|_F$ by
\beqn
\|X\otimes Y\otimes (Z - \tilde Z)\|_F + \|X\otimes (Y - \tilde Y)\otimes \tilde Z\|_F + \|(X - \tilde X)\otimes \tilde Y \otimes \tilde Z\|_F.
\eeqn
Since for any $A,B$, $\|A\otimes B\|_F = \|A\|_F\,\|B\|_F$, the quantity above is less than~$3\eps$.
\end{proof}

\subsection{Deviation bounds}

In this section we collect some useful large deviation bounds.

\begin{fact}[Gaussian tail bound]\label{fact:gausstail} Let $g\sim\normal$ be a standard normal random variable. Then for any $t\geq 0$, 
$$\Pr\big[ |g| \geq t \big] \,\leq\, 2e^{-t^2/2}.$$
\end{fact}

\begin{fact}[Hoeffding's inequality]\label{fact:hoeffding} Let $h_1,\ldots,h_N$ be independent centered random variables such that for every $i\in [N]$, we have $\Pr\big[h_i \in [a_i,b_i]\big]=1$ . Then for any $t\geq 0$, 
$$\Pr\left[ \Big|\sum_{i=1}^N h_i  \Big| \geq t \right] \,\leq\, 2e^{-2t^2/\sum_i (b_i-a_i)^2}.$$
\end{fact}

\begin{fact}[Bernstein's inequality, see eg. Prop.~16 in~\cite{Versh_matrices}]\label{fact:bernstein} Let $h_1,\ldots,h_N$ be independent centered random variables and $K>0$ be such that $\Pr\big[|h_i| \geq t\big] \leq e^{1-t/K}$ for all $i$ and $t\geq 0$. Then for any $a\in \R^N$ and $t\geq 0$,
$$\Pr\left[\Big|\sum_{i=1}^N a_i h_i\Big| \geq t\right] \leq 2e^{-\frac{1}{4e} \min\big( \frac{t^2}{2eK^2 \|a\|_2^2}, \frac{t}{K \|a\|_\infty}\big)}.$$
\end{fact}

\begin{corollary}[$\chi^2$ tail bound]\label{cor:chisquare} Let $\ket{g}$ be a random vector distributed according to $\normal^{N}$. Then for every $t \geq 0$,
$$ \Pr \Big[ \big| \|\ket{g}\|_2^2 - N \big| \geq t \Big] \,\leq\, 2\,e^{-\frac{1}{8e} \min\big( \frac{t^2}{4eN}, t\big)}.$$
\end{corollary}

\begin{proof} 
Write $\ket{g} = g_1\ket{1} + \cdots + g_N\ket{N}$ where $g_1,\dots,g_N$ are i.i.d. standard normal random variables. By Fact~\ref{fact:gausstail}, for every $i$ the $g_i$ satisfy that for every $t \geq 0$,
\begin{align*}
\Pr\big[ |g_i^2-1| \geq t\big] &= \Pr\big[ g_i^2 \geq t+1 \big] + \Pr\big[ g_i^2 \leq 1-t\big]\\
&\leq e e^{-(t+1)/2},
 \end{align*}
 where the factor $e$ in front ensures that the bound is trivial whenever the second term $ \Pr( g_i^2 \leq 1-t)$ is nonzero. Hence the random variables $h_i:=g_i^2-1$ satisfy the hypothesis of Fact~\ref{fact:bernstein} with $K=2$, which immediately gives the claimed bound. 
\end{proof}

\begin{corollary}[Projections of Bernoulli vectors]\label{cor:berntail} Let $\eps_{ij}$, $i,j\in[N]$ be $i.i.d.$ Bernoulli random variables, and $a\in \R^N$. Then
$$\Pr\left[ \Big| \sum_{j=1}^N \Big( \sum_{i=1}^N a_i \eps_{ij} \Big)^2 - N\|a\|_2^2 \Big| > t \right] \,\leq\,2e^{-\frac{1}{4e} \min\big( \frac{t^2}{8e \|a\|_2^4 N}, \frac{t}{2\|a\|_2^2 }\big)}.$$
\end{corollary}

\begin{proof}
For any $j\in [N]$ let $\eta_j =  \big( \sum_{i=1}^N a_i \eps_{ij} \big)^2 - \|a\|_2^2$. The $\eta_j$ are independent centered random variables, and by Fact~\ref{fact:hoeffding} they satisfy a tail bound as required by Fact~\ref{fact:bernstein}, with $K= 2\|a\|_2^2$. The corollary follows.  
\end{proof}

The following is a special case of a result due to Lata\l a (see Corollary~1 in~\cite{Lataa2006}).

\begin{corollary}\label{cor:gaussherm} Let $A\in\herm{N}$ be a Hermitian matrix and $\ket{g}\sim\normal^N$. 
Then, for any $t\geq 0$,
$$\Pr\Big[ \big|\bra{g}A\ket{g} - \Tr(A)\big|  \geq  t\Big] \,\leq\, 2\,e^{-\frac{1}{24 e}\min\big(\frac{t^2}{12e\|A\|_F^2},\frac{t}{\|A\|_\infty}\big)}. $$
\end{corollary}

\begin{proof}
Since $A$ is Hermitian, it is unitarily diagonalizable: $A = UDU^\dagger$ where $D = \text{diag}(\lambda_i)$, and the $\lambda_i$ are its real eigenvalues. Then 
$$ \bra{g}A\ket{g} = \sum_{i=1}^N \lambda_i |\bra{i} U\ket{g}|^2,$$
where the $g_i$ are the standard normal distributed coefficients of the random vector $\ket{g}$. Since the rows of $U$ are orthogonal, the $\bra{i}U\ket{g}$ are independent random variables. Moreover, since $\ket{g}$ is real, we have $|\bra{i}U\ket{g}|^2 = \big(\Re(\bra{i}U)\ket{g}\big)^2 + \big(\Im(\bra{i}U)\ket{g}\big)^2$, where $\Re(\bra{i}U)$ and $\Im(\bra{i}U)$ are the real and imaginary parts of the unit vector $\bra{i}U$ forming the $i$-th row of $U$. By rotation invariance, we have that for arbitrary $\ket{x}\in\R^N$, the random variable $\braket{x}{g}$ is distributed as $\mathcal{N}\!(0,\|\ket{x}\|^2)$. It follows from Fact~\ref{fact:gausstail} that for every $i\in[N]$, we have
\begin{align*}
\Pr\Big[|\bra{i}U\ket{g}|^2 \geq t\Big] & \leq \Pr\Big[ \big(\Re(\bra{i}U)\ket{g}\big)^2 \geq t/2\Big] + \Pr\Big[\big(\Re(\bra{i}U)\ket{g}\big)^2 \geq t/2\Big]\\
&\leq 2e^{-t/(4\|\Re(\bra{i}U\|)^2)}+2e^{-t/(4\|\Im (\bra{i}U_i)\|^2)}\\
&\leq 4e^{-t/4}
\end{align*}
Hence we can apply Fact~\ref{fact:bernstein} with $K = 4(\ln(4/e)+1)\leq 6$ to obtain for any $t\geq 0$:
$$\Pr\left[ \Big|\sum_{i=1}^N \lambda_i \Big(|\bra{i}U\ket{g}|^2 - \Exp\big[|\bra{i}U\ket{g}|^2\big]\Big)\Big| \geq t \right]  \,\leq\, 2\,e^{-\frac{1}{24e} \min\big(\frac{t^2}{12e\|A\|_F^2},\frac{t}{\|A\|_\infty}\big)} $$
which proves the claim since $\Exp\big[|\bra{i}U\ket{g}|^2\big] =1$ for every $i\in[N]$ and $\sum_{i=1}^N \lambda_i = \Tr(A)$. 
\end{proof}

Finally, we state without proof an analogue of the preceding Corollary which applies to Bernoulli random variables, and is a special case of a result of Hanson and Wright~\cite{HW71}.

\begin{theorem}\label{thm:hw} 
There exists a constant $D>0$ such that the following holds. Let $A\in\herm{N}$ be a Hermitian matrix  and $\eps_{ij}$ i.i.d. Bernoulli random variables. Then, for any $t\geq 0$,
$$\Pr\left[ \Big|\sum_{i,j} A_{ij} \eps_{ij} - \Tr(A)\Big|\geq t\right] \,\leq\, 2\,e^{-C \min\big(\frac{t^2}{\|A\|_F^2},\frac{t}{\|A\|_\infty}\big)}. $$
\end{theorem}

\subsection{Grothendieck's inequality}
\label{sec:grothendieck}

We use the following version of Grothendieck's inequality~\cite{Grothendieck:1953}. The bounds on the constants involved come from~\cite{Haagerup:1987} and~\cite{BMMN11}.

\begin{theorem}[Grothendieck's inequality]
There exists a universal constant $K_G^\R<1.783$ such that the following holds. Let $N$ and $d$ be positive integers. Then, for any matrix $M \in \mat{N}$ with real coefficients and any complex unit vectors $x_1,\dots,x_N$, $y_1,\dots,y_N\in \sphere{d}$, we have
\beq\label{eq:grothineq}
\Big| \sum_{i,j=1}^N M_{ij} \langle x_i,y_j\rangle\Big| \leq K_G^\R\max_{\chi,\upsilon:[N]\to\pmset{} }\sum_{i,j=1}^N M_{ij} \chi(i)\upsilon(j),
\eeq
If we allow $\chi,\upsilon$ on the right-hand side of~\eqref{eq:grothineq} to take values in the set of all complex numbers with modulus (at most) $1$, then the constant $K_G^\R$ may be replaced by the complex Grothendieck constant $K_G^\C < 1.405$.  
\end{theorem}


\section{Unbounded gaps}\label{sec:unbounded}

This section is devoted to the proof of Theorem~\ref{thm:main}. The theorem is proved in two steps. In the first step we associate a three player XOR game $G$ to any $3$-tensor $T$, and relate the quantum-classical gap for that game to spectral properties of $T$. We emphasize that the game $G=G(T)$ is \emph{not} defined from $T$ in the most straightforward way (using $T$ as the game tensor), but through a more delicate transformation, based on the use of the Fourier transform, which is exposed in Section~\ref{sec:pauligames}.

\begin{proposition}\label{prop:pauligame} Let $n$ be an integer and let $N = 2^{n}$. Let $T$ be any $3$-tensor of dimensions $N^2\times N^2\times N^2$. Then there exists a $3$-player XOR game $G=G(T)$ such that 
$$\frac{\beta^*(G)}{\beta(G)} \geq \frac{1}{4N^{3/2}} \frac{\|T\|_{3,3}}{\|T\|_{2,2,2}}.$$
Moreover, in the game $G$ there are $N^2$ questions to each player, and there is a entangled strategy which achieves the claimed violation and uses only $N$-dimensional Pauli observables. 
\end{proposition}

In the second step we show the existence of a tensor $T$ such that $\|T\|_{3,3}/\|T\|_{2,2,2}$ is large. 

\begin{proposition}\label{prop:goodtensor}
There is a constant $C>0$ such that for any integer $N$ there exists a $3$-tensor $T$ of dimensions $N^2\times N^2\times N^2$ such that 
$$\frac{\|T\|_{3,3}}{\|T\|_{2,2,2}}\, \geq\, CN^2 \log^{-5/2} N.$$ 
\end{proposition} 

Theorem~\ref{thm:main} trivially follows from the two propositions above. While we have not made the constants in the preceding propositions completely explicit, it is not hard to extract numerical values from our proofs; in particular we give precise estimates for all our probabilistic arguments. 
Proposition~\ref{prop:pauligame} is proved in Section~\ref{sec:pauligames}, and Proposition~\ref{prop:goodtensor} is proved in Section~\ref{sec:constructT}. 

\subsection{Pauli XOR games}\label{sec:pauligames}

Let $T$ be a complex $3$-tensor of dimensions $N^2\times N^2\times N^2$, where $N=2^{n}$ and $n$ is an arbitrary integer. Based on $T$ we define a three-player XOR game $G=G(T)$ with the following properties:
\begin{enumerate}
\item There are $N^2$ questions per player,
\item The best classical strategy for game $G(T)$ achieves a bias of at most $N^{9/2}\|T\|_{2,2,2}$,
\item There is a entangled strategy which uses only Pauli matrices as observables and entanglement of local dimension $N$ per player and achieves a bias of at least $ (N^3/4)\|T\|_{3,3}$.
\end{enumerate}
Properties 2. and 3. imply that in game $G(T)$, the ratio between the entangled and classical biases is at least
$$ \frac{\beta^*(G)}{\beta(G)} \,\geq\, \frac{1}{4N^{3/2}} \frac{ \|T\|_{3,3}}{\|T\|_{2,2,2}}, $$
proving Proposition~\ref{prop:pauligame}. 

\medskip

Let $T$ be a $N^2\times N^2\times N^2$ tensor. By replacing $T$ by either $(T+T^\dagger)/2$ or $i(T-T^\dagger)/2$, we may assume that $T$, when seen as an $N^3\times N^3$ matrix, is also Hermitian. One of these two possible choices necessarily results in a ratio of the $\|\cdot \|_{3,3}$ norm to the $\|\cdot \|_{2,2,2}$ norm that is at least half of what it was for $T$. In order to associate an XOR game to $T$, we first define (possibly complex) coefficients indexed by Pauli matrices $P,Q,R\in \mathcal{P}_{n}$ as follows 
$$M_{P,Q,R}\,:=\,\widehat{T}(P,Q,R)\,=\, \sum_{(i,i'),(j,j'),(k,k')\in [N]^2} T_{(i,i'),(j,j'),(k,k')} P_{i,i'} Q_{j,j'} R_{k,k'}.$$ 
In order to obtain an XOR game $G=G(T)$, we take either the real or the imaginary part of the coefficients $M_{P,Q,R}$ (whichever allows for the largest entangled bias), and normalize the resulting sequence according to its $\ell_1$ norm (note that this normalization has no effect on the ratio of the biases that is considered in Proposition~\ref{prop:pauligame}). This results in a game with $N^2$ questions per player, indexed by the Pauli matrices. Since we are ultimately only concerned with the ratio $\|T\|_{3,3}/\|T\|_{2,2,2}$, without loss of generality we assume that the two transformations made above (making $T$ Hermitian and such that the coefficients defined above are all real) resulted in the $\|\cdot \|_{3,3,3}$ norm being divided by a factor at most $4$, and the $\|\cdot\|_{2,2}$ norm remaining unchanged. 

The fact that property 1. above holds is clear, by definition. Next we prove that property 2. holds. Let $\chi,\upsilon,\zeta:\pauli{n}\to\pmset{}$ be an optimal classical strategy. Define the matrices $X = \sum_{P\in\mathcal{P}_{n}}\chi(P) P$, $Y = \sum_{Q\in\mathcal{P}_{n}} \upsilon(Q) \,Q$ and $Z = \sum_{R\in\mathcal{P}_{n}} \zeta(R)\,R$. Then $X,Y$ and $Z$ are Hermitian, and 
$$ \|X\|_F^2\,=\,\Tr(X^\dagger X) \,=\,\sum_{P,P'\in\pauli{n}} \chi(P) \chi(P')\Tr(P^\dagger P') \,=\, N \sum_{P\in\pauli{n}} \chi(P)^2 \,=\, N^3,$$
and the same holds for $Y$ and $Z$. By the Cauchy-Schwarz inequality, the classical bias can be bounded as
\begin{align*}
 \beta(M) &= \sum_{P,Q,R}  \widehat{T}(P,Q,R) \,\chi(P)\upsilon(Q)\zeta(R)\\
  &= \sum_{P,Q,R\in\pauli{n}} \langle T,X\otimes Y\otimes Z\rangle \,\chi(P)\upsilon(Q)\zeta(R)\\
  &\leq \max_{X,Y,Z\in\mathcal{B}(\herm{N},N^{3/2})}\, \langle T,X\otimes Y \otimes Z \rangle \\
 &\leq N^{9/2}  \|T\|_{2,2,2}.
 \end{align*}

Finally, we prove property 3. by exhibiting a good entangled strategy for $G(T)$. We simply let the observable corresponding to question $P$ (resp. $Q,R$) be the $n$-qubit Pauli matrix $P$ (resp. $Q,R$). Let $\ket{\Psi}$ be a shared entangled state. The bias of the corresponding strategy is
$$
\sum_{P,Q,R} \widehat{T}(P,Q,R) \,\bra{\Psi} P\otimes Q \otimes R \ket{\Psi} \,=\, N^3 \bra{\Psi} T \ket{\Psi} \,=\, N^3 \|T\|_{3,3}, 
$$
where for the last equality we chose $\ket{\Psi}$ an eigenvector of $T$ with largest eigenvalue.  

\begin{remark} In our construction, the only properties of the Pauli matrices that we use is that they form a family of observables that is orthogonal with respect to the Hilbert-Schmidt inner product on $\herm{N}$. Any other such family would lead to a completely analogous construction (in which the player's observables in the entangled strategy are replaced by the corresponding elements).
\end{remark}

\subsection{Constructing a good tensor $T$}\label{sec:constructT}

In this section we prove Proposition~\ref{prop:goodtensor} by giving a probabilistic argument for the existence of a tensor $T$ with good spectral properties. Let $N$ be an integer, and $\ket{g}$ the (random) $N^3$-dimensional vector
\beqn
\ket{g} := \sum_{i,j,k=1}^Ng_{ijk}\ket{i}\ket{j}\ket{k} \sim\normal^{N^3},
\eeqn
where the $g_{ijk}$ are i.i.d. $\normal$ random variables.
We define a tensor $T$ depending on the $g_{ijk}$, and then prove bounds on the $\|\cdot\|_{3,3}$ and $\|\cdot\|_{2,2,2}$ norms of $T$ that hold with high probability over the choice of the $g_{ijk}$. 
Let 
\begin{align}
T &:= \sum_{i\neq i',j\neq j',k\neq k'} g_{ijk}\, g_{i'j'k'} \ketbra{i,j,k}{i',j',k'}.\label{eq:deft}
\end{align}
$T$ is a real $N^3\times N^3$ symmetric matrix that equals $\ketbra{g}{g}$ with some coefficients zeroed out, including those on the diagonal. Hence $T$ is very close to a rank $1$ matrix and it should therefore be no surprise that its spectral norm is large, as we show in Section~\ref{sec:Tlb} below. More work is needed to upper bound the $\|\cdot\|_{2,2,2}$ norm of $T$. In particular, we note that zeroing out the diagonal coefficients is essential to getting a good bound on $\|T\|_{2,2,2}$. While we show in Section~\ref{sec:Tub} that with high probability over $\ket{g}$ we have $\|T\|_{2,2,2} = O(N\log^{5/2}N)$, it is not hard to see that in expectation we already have $\|\ketbra{g}{g}\|_{2,2,2} = \Omega(N\sqrt{N})$ (indeed, simply choose $X=Y=Z = I/\sqrt{N}$ in the definition of $\|\cdot\|_{2,2,2}$). Zeroing out some entries of $\ketbra{g}{g}$  approximately preserves the spectral norm, but decreases its norm as a trilinear operator by almost a factor~$\sqrt{N}$.

\begin{remark} The same construction, with the normal random variables $g_{ijk}$ replaced by i.i.d. Bernoulli random variables, can be used to obtain similar results.\footnote{We thank Ignacio Villanueva for asking this question.} Indeed, Lemma~\ref{lem:tensorlb} below holds trivially in that case, and to obtain the analogue of Lemma~\ref{lem:tensorub} it suffices to replace the use of Corollary~\ref{cor:chisquare} and Corollary~\ref{cor:gaussherm} in the proof of Lemma~\ref{lem:netconc} by Corollary~\ref{cor:berntail} and Theorem~\ref{thm:hw} respectively. 
\end{remark}

\subsubsection{A lower bound on the spectral norm}\label{sec:Tlb}

A lower-bound on the spectral norm of $T$ as defined in~\eqref{eq:deft} follows easily from the fact that it is, by definition, very close to a rank-$1$ matrix. 
We show the following.

\begin{lemma}\label{lem:tensorlb}
For any $\tau>0$ and all large enough $N$ it holds that 
$$\|T\|_{3,3} \,\geq\, N^3 - \tau N^2$$
with probability at least $1-e^{-\Omega(\tau^2)}$.
\end{lemma} 

\begin{proof}
 Define $\ket{\Psi} = N^{-3/2}\ket{g}$. By Corollary~\ref{cor:chisquare}, for any $\delta>0$ we have that
 \beqn
 \Pr\Big[\sum_{i,j,k} g_{ijk}^2 \leq (1-\delta)N^3 \Big] \leq 2e^{-\delta^2 N^3/(64e^2)}.
 \eeqn
Provided this holds,
\beq\label{eq:tensor-1}
\|\ket{\Psi}\|^2 = \frac{1}{N^3}\sum_{i,j,k} g_{ijk}^2 \leq 1-\delta.
\eeq
Another application of Corollary~\ref{cor:chisquare}, together with a union bound, shows that the probability that there exists an $i\in[N]$ such that
 $ \sum_{j,k}  g_{ijk}^2 \geq (1+\delta)N^2 $ is at most $2Ne^{-\delta^2 N^2/(64e^2)}$. Provided this holds, 
\beq\label{eq:tensor-2}
 \sum_{i}\Big(\sum_{j,k} |g_{ijk}|^2\Big)^2 \,\leq\,(1+\delta)^2 N^5
\eeq
and the same holds symmetrically for $j$ or $k$. This lets us bound
\begin{align*}
\bra{\Psi} T \ket{\Psi} &= \frac{1}{N^3} \sum_{i\neq i',j\neq j',k\neq k'} |g_{ijk}|^2 \,|g_{i'j'k'}|^2\\
&\geq \frac{1}{N^3} \Big(\Big(\sum_{i,j,k} |g_{ijk}|^2 \Big)^2 - \sum_{i}\Big(\sum_{j,k} |g_{ijk}|^2\Big)^2-\sum_{j}\Big(\sum_{i,k} |g_{ijk}|^2\Big)^2-\sum_{k}\Big(\sum_{i,j} |g_{ijk}|^2\Big)^2\Big)\\
&\geq \frac{(1-\delta)^2 N^6-3(1+\delta)^2 N^5 }{N^3} \geq (1-3\delta)\,N^3,
\end{align*} 
where the second inequality uses~\eqref{eq:tensor-1} and~\eqref{eq:tensor-2}, and the last holds for large enough $N$. Hence, using~\eqref{eq:tensor-1} once more,
$$\|T\|_{3,3} \,\geq\, \frac{\bra{\Psi} T \ket{\Psi}}{\|\ket{\Psi}\|^2}\,\geq\, (1-3\delta)\,N^3 (1-\delta)^{-1} \,\geq\, (1-6\delta)N^3$$
for small enough $\delta$. The claimed bound follows by setting $\delta = \tau/(6N)$.
\end{proof}

\subsubsection{Upper-bounding $\|T\|_{2,2,2}$}\label{sec:Tub}

In this section we give an upper bound for $\|T\|_{2,2,2}$ that holds with good probability over the choice of $T$, where $T$ is as in~\eqref{eq:deft}, a $3$-tensor of dimensions $N^2\times N^2\times N^2$. 
Recall that
\beqn
\|T\|_{2,2,2} = \max_{X,Y,Z\in\hball{N}} |\langle T,X\otimes Y \otimes Z\rangle|.
\eeqn
We prove the following.  

\begin{lemma}\label{lem:tensorub} There exist universal constants $d,D>0$ such that for all large enough $N$, we have
$$ \|T\|_{2,2,2} \leq  D N (\ln N)^{5/2} $$
with probability at least $1-e^{-d N}$ over the choice of $\ket{g}$.
\end{lemma}

\noindent
We note that if $T$ was a random tensor with entries i.i.d. standard normal, then a result by Nguyen et al.~\cite{NDT10} would show that $\|T\|_{2,2,2} = O\big(N\sqrt{\log N}\big)$ holds with high probability. However, the entries of our tensor $T$ are not independent, and we need to prove a bound tailored to our specific setting.

Our first step consists in showing that the supremum in the definition of $\|T\|_{2,2,2}$ can be restricted to a supremum over projector matrices, at the cost of the loss of a logarithmic factor in the bound.\footnote{We thank Gilles Pisier for suggesting the use of this decomposition.}

\begin{lemma}\label{lem:projbound}
Let $\ket{g}$ be a vector in~$\R^{N^3}$ and let $T$ be the associated tensor, as in~\eqref{eq:deft}. Then
\beq\label{eq:sup-proj}
\|T\|_{2,2,2} \leq 64\, (\ln N)^{3/2}\: \max\: \big| \bra{g}  X\otimes Y\otimes Z \ket{g} - \Tr(X\otimes Y\otimes Z)\big|,
\eeq
where the maximum is taken over all triples $(X,Y,Z)\in\projnorm{N}^3$.
\end{lemma}


\begin{proof}
Let $X,Y,Z\in B\big(\herm{N}\big)$ be traceless Hermitian matrices such that
\beqn
\|T\|_{2,2,2} = \langle T,X\otimes Y\otimes Z\rangle = \bra{g}X\otimes Y\otimes Z\ket{g},
\eeqn
where the second equality follows from the definition of $T$. 
Decompose $X,Y,Z$ as per Lemma~\ref{lem:lorentz}, giving
$$
X = \sum_{x}\alpha_x X_x,\qquad Y = \sum_{y}\beta_y  Y_y\qquad\text{and}\qquad Z = \sum_{z}\gamma_z  Z_z,
$$
where $\|(\alpha_x)_x\|_1, \|(\beta_y)_y\|_1, \|(\gamma_z)_z\|_1\leq 4\sqrt{\ln N}$ and $X_x,Y_y,Z_z \in \projnorm{N}$.
Note that
\beqn
0 = \Tr(X\otimes Y\otimes Z) = \sum_{x,y,z}\alpha_x\beta_y\gamma_z\,\Tr(X_x\otimes Y_y\otimes Z_z).
\eeqn
\noindent By linearity and H\"{o}lder's inequality, we have
\begin{align*}
 \bra{g}X\otimes Y\otimes Z\ket{g} &- \Tr(X\otimes Y\otimes Z)\\ 
 &= \sum_{x,y,z}\alpha_x\beta_y\gamma_z\, \Big(\bra{g}X_x\otimes Y_y\otimes Z_z\ket{g} - \Tr(X_x\otimes Y_y\otimes Z_z)\Big)\\
 &\leq 64 (\ln N)^{3/2}\max_{x,y,z} \: \big| \bra{g}  X_x\otimes Y_y\otimes Z_z \ket{g} - \Tr(X_x\otimes Y_y\otimes Z_z)\big|,
\end{align*}
 proving the lemma.
\end{proof}

Our next step is to show that we may further restrict the maximum on the right-hand side of~\eqref{eq:sup-proj} to a maximum over projectors taken from the $\eps$-net $\net_\eps$ given in Definition~\ref{def:triplenet}.


\begin{lemma}\label{lem:netbound}
Let $\ket{g}$ be a vector in~$\R^{N^3}$, $T$ the associated tensor and $\eps>0$. Then
\beq\label{eq:sup-net}
\|T\|_{2,2,2} \leq 64\, (\ln N)^{3/2}\:\Big( \max\: \big| \bra{g}  X\otimes Y\otimes Z \ket{g} - \Tr(X\otimes Y\otimes Z)\big| +3\eps\,\big(N^{3/2}+\big\|\ket{g}\big\|_2^2\big)\Big),
\eeq
where the maximum is taken over all  $X\otimes Y\otimes Z\in \net_{\eps} $.
\end{lemma}

\begin{proof} Fix a triple $(X,Y,Z)\in \projnorm{N}^3$. By Proposition~\ref{prop:triplenet}, there exists an $\tilde{X}\otimes\tilde{Y}\otimes\tilde{Z}\in\net_\eps$ such that
$$ \big\| X \otimes Y \otimes Z - \tilde{X}\otimes\tilde{Y}\otimes\tilde{Z} \big\|_F\,\leq\,3\,\eps.$$
By the Cauchy-Schwarz inequality, we have
\beqrn
|\bra{g} X\otimes Y\otimes Z\ket{g} - \bra{g} \tilde X\otimes \tilde Y\otimes \tilde Z\ket{g}| &=& |\bra{g}  X\otimes Y\otimes Z- \tilde X\otimes \tilde Y\otimes \tilde Z\ket{g}|\nonumber\\[.1cm]
 &\leq&  \| X\otimes Y\otimes Z- \tilde X\otimes \tilde Y\otimes \tilde Z\|_F\: \|\ketbra{g}{g}\|_F.
\eeqrn
Another application of the Cauchy-Schwarz inequality and the definition of the Frobenius norm give
\beqrn
\big| \Tr(X\otimes Y\otimes Z -  \tilde X\otimes \tilde Y\otimes \tilde Z) \big| &=& |\langle I, X\otimes Y\otimes Z -  \tilde X\otimes \tilde Y\otimes \tilde Z\rangle|\\
&\leq& N^{3/2}\, \| X\otimes Y\otimes Z- \tilde X\otimes \tilde Y\otimes \tilde Z\|_F.
\eeqrn
Hence the lemma follows from Lemma~\ref{lem:projbound}.
\end{proof}

We upper-bound the right-hand side of~\eqref{eq:sup-net} by first showing that for any fixed triple $(k,\ell,m)\in [N]^3$ and $X\otimes Y\otimes Z\in \net_\eps^{(k,\ell,m)}$, this quantity is bounded with high probability over the choice of $\ket{g}$. We conclude by applying a union bound over the net $\net_{\eps} = \bigcup_{(k,\ell,m)\in [N]^3}\net_\eps^{(k,\ell,m)}$.

\begin{lemma}\label{lem:netconc}
There exist constants $C,c>0$ such the following holds. 
For any $0<\eps\leq N^{-3}$ and $\tau \geq CN\ln(1/\eps)$, the probability over the choice of~$\ket{g}$ that there exists an $X\otimes Y\otimes Z\in \net_{\eps} $ such that
\beq\label{eq:netconc}
 \big|\bra{g}X\otimes Y\otimes Z\ket{g} - \Tr(X\otimes Y\otimes Z)\big| > \tau
\eeq
is at most $e^{-c\tau}$.
\end{lemma}

\begin{proof} Fix a triple $(k,\ell,m)\in [N]^3$, and assume that $k\geq \max\{\ell, m\}$, the other cases being reduced to this one by permutation of the indices. Since $k+\ell + m \leq 3k$, we have
\beq\label{eq:netconc-0}
\big|\net_\eps^{(k,\ell,m)}\big| \leq 8\left(\frac{5}{\eps}\right)^{(k+\ell+m)N} \leq e^{3kN\ln(5/\eps)+3}.
\eeq
We distinguish two cases.

Case 1: $\ell m >k$. Fix an  $X\otimes Y\otimes Z\in\net_\eps^{(k,\ell,m)}$.
By definition of the nets $\net_\eps^j$,
$$\| X\otimes Y\otimes Z\|_F \leq 1 \quad\text{and}\quad\| X\otimes Y\otimes Z\|_{\infty} = \frac{1}{\sqrt{k\ell m}}.$$
Hence, by Corollary~\ref{cor:gaussherm} there exists a constant $c'>0$ such that for any $\tau>0$
\beq\label{eq:netconc2}
\Pr_{\ket{g}}\Big[ \big|\bra{g}X\otimes Y\otimes Z\ket{g} - \Tr(X\otimes Y\otimes Z)\big| \geq \tau\Big] \leq e^{-c'\min\left\{\tau^2, \tau\sqrt{k\ell m} \right\}}.
\eeq
Our assumption $\ell m> k$ implies $\sqrt{k\ell m} > k$, hence the probability above is at most $e^{-c'\min\{\tau^2, k\tau\}}$. Using the bound~\eqref{eq:netconc-0} on the size of $\net_\eps^{(k,\ell,m)}$, by a union bound there exists a $C'>0$ such that for any $\tau\geq C'\,N\ln (1/\eps)$ the probability that there exists an $X'\otimes Y'\otimes Z'\in\net_\eps^{(k,\ell,m)}$ such that
\beqn
\big|\bra{g}X'\otimes Y'\otimes Z'\ket{g} - \Tr(X'\otimes Y'\otimes Z')\big| \geq \tau
\eeqn
is at most $e^{-\Omega(\tau)}$.
\medskip

Case 2: $k\geq \ell m$. 
Fix an $X\otimes Y\otimes Z\in\net_\eps^{(k,\ell,m)}$.
Since $X$, $Y$ and $Z$ are normalized projectors, 
\beqn
\Tr(X\otimes Y\otimes Z) \leq \sqrt{k\ell m} \leq k \leq N.
\eeqn
Write the spectral decompositions of $X$, $Y$ and $Z$ as
\begin{align*}
X = \frac{1}{\sqrt{k}}\sum_{p} \ketbra{x_p}{x_p},& &Y = \frac{1}{\sqrt{\ell}}\sum_{q}  \ketbra{y_q}{y_q}& &\text{and}& & Z = \frac{1}{\sqrt{m}}\sum_{r} \ketbra{z_r}{z_r},
\end{align*}
where the indices $p,q,r$ run from 1 to at most $k,\ell,m$, respectively.
For any unit $\ket{y},\ket{z}\in\C^N$, define the $N$-dimensional vector $\ket{g(y,z)} = (I\otimes \bra{y}\otimes \bra {z})\ket{g}$. By rotation invariance of the Gaussian distribution, $\ket{g(y,z)}$ is distributed according to~$\normal^N$.
Since  $\ket{x_1},\ket{x_2},\dots$ are pairwise orthogonal, we have
\begin{align}
|\bra{g}X\otimes Y\otimes Z\ket{g}| &= \frac{1}{\sqrt{k\ell m}}\sum_{p,q,r} \big|\braket{x_p}{g(y_q,z_r)}\big|^2 \notag\\
&\leq \sqrt{\frac{\ell m}{k}} \max_{\ket{y},\ket{z}\in\C^N,\,\|\ket{y}\|,\|\ket{z}\|\leq 1} \big\|\ket{g(y,z)} \big\|_2^2\notag\\
&\leq \max_{\ket{y},\ket{z}\in S_\eps} \big\|\ket{g(y,z)} \big\|_2^2 + 4\,\eps \|\ket{g}\|_2^2,\label{eq:netconc-1}
\end{align}
where for the last inequality we used that $\sqrt{\ell m/k} \leq 1$ (which follows from our assumption $k\geq \ell m$), and that for any unit $y,\tilde{y},z,\tilde{z}$,
\begin{align*}
\Big| \big\|\ket{g(y,z)} \big\|_2^2 - \big\|\ket{g(\tilde{y},\tilde{z})} \big\|_2^2 \Big| &=\Big| \big( \bra{g(y,z)}+\bra{g(\tilde{y},z)} \big)\big( \ket{g(y,z)}-\ket{g(\tilde{y},z)}\big) \\
&\qquad + \big( \bra{g(\tilde{y},z)}+\bra{g(\tilde{y},\tilde{z})} \big)\big( \ket{g(\tilde{y},z)}-\ket{g(\tilde{y},\tilde{z})}\big) \Big|\\
&\leq \big\|\ket{g(y+\tilde{y},z)} \big\|_2 \big\|\ket{g(y-\tilde{y},z)}\big\|_2 +  \big\|\ket{g(\tilde{y},z+\tilde{z})} \big\|_2 \big\|\ket{g(\tilde{y},z-\tilde{z})}\big\|_2\\
&\leq 2 \big\| \ket{g} \big\|_2^2 \big( \big\|y-\tilde{y}\big\|_2 + \big\|z-\tilde{z}\big\|_2\big).
\end{align*}
Applying Corollary~\ref{cor:chisquare}, there exists a $c''>0$ such that for any $\tau>0$ the maximum in~\eqref{eq:netconc-1} is greater than $N + \tau$ with probability at most $e^{-c''\min\{\tau^2/N,\tau\}}$. Since by Fact~\ref{fact:epsnet} $|S_\eps| \leq e^{-2\ln(1/\eps)N}$, a union bound shows that there exists a $C''>0$ such that for all $\tau \geq C''N\ln(1/\eps)$ the bound
$$ |\bra{g}X\otimes Y\otimes Z\ket{g}| \leq \eps (N^{3} + N\tau) + \tau $$
holds with probability at least $e^{-c'''\tau}$ over the choice of $\ket{g}$, for some $c'''>0$. (Here we again used Corollary~\ref{cor:chisquare} to upper-bound $\|\ket{g}\|_2^2 \leq N^3+N\tau$ with probability at least $1-e^{-\Omega(\tau)}$.)

\medskip

The lemma follows for some $c,C>0$ by combining the two cases analyzed above and performing a union bound over all $N^3$ triples $(k,\ell,m)$.
\end{proof}

We are now in a position to prove Lemma~\ref{lem:tensorub}.

\begin{proof}[of Lemma~\ref{lem:tensorub}]
Let $\eps = N^{-3}$ and $\tau = C N\ln(1/\eps)$, where $C$ is the constant appearing in the statement of Lemma~\ref{lem:netconc}. That lemma shows that the bound 
$$ \big|\bra{g}X\otimes Y\otimes Z\ket{g} - \Tr(X\otimes Y\otimes Z)\big| \leq CN\ln(1/\eps)$$
holds except with probability at most $e^{-cCN\ln(1/\eps)}$. Moreover, by Corollary~\ref{cor:chisquare}, there is a $C'>0$ such that
\beqn
3\eps\, \Big(N^{3/2} + \big\|\ket{g}\big\|_2^2\Big) \leq 6\eps N^{3/2},
\eeqn
except with probability at most $1-e^{-C'N}$. Combining these two bounds with the estimate of Lemma~\ref{lem:netbound} proves the lemma, provided $d$ is chosen small enough and $D$ large enough. 
\end{proof}


\section{Upper bounds on violations}\label{sec:ub}

\subsection{Bounds in terms of the number of questions}

In this section we prove Theorem~\ref{thm:questionbound}, which we restate here for convenience.

\begin{theorem2}\label{thm:questionbound2}
For any $3$-player XOR game $G$ in which there are at most $Q$ possible questions to the third player,
 $$\beta^*(G) \,\leq\, \sqrt{Q}\,K_G^\R\,\beta(G),$$
 where $K_G^\R <1.783$ is the real Grothendieck constant.
\end{theorem2}

The two main ingredients in the proof are a useful technique of Paulsen and Grothendieck's inequality. Paulsen's technique (see~\cite[Proposition 2.10]{paulsen:1992}) lets us ``decouple'' the third player from the other two players and turn his part of the entangled strategy into a classical one at a loss of a factor $\sqrt{Q}$ in the overall bias.\footnote{This technique is based on so-called Rademacher averaging, a well-know method in the field of Banach spaces.}
Slightly more precisely, the proof goes as follows.
By grouping the game tensor and the observables of the first two players together, the entangled bias takes the form
\beqn
\beta^*(G)\,=\,\psibra \sum_{k=1}^Q M_k\otimes C_k\psiket,
\eeqn
where the $C_k$ are the third player's observables in an optimal entangled strategy.
The decoupling technique relies on a collection of i.i.d. $\pmset{}$-valued symmetrically distributed Bernoulli random variables $\eps_1,\dots,\eps_Q$ which are used to split the above sum into two sums. Using the fact that $\Exp[\eps_k\eps_\ell] = \delta_{k\ell}$, the above expression can be written as
\beqn
\Exp\left[\Big(\psibra\sum_{k=1}^Q M_k\otimes (\eps_kI) \Big)\Big(\sum_{\ell=1}^Q\eps_{\ell} I\otimes C_\ell\psiket\Big)\right].
\eeqn
After two applications of the Cauchy-Schwarz inequality, the third player's classical strategy will be a certain instantiation of the random variables $\eps_k$ appearing in the left brackets, while the factor $\sqrt{Q}$ will come from the term between the right brackets. An application of Grothendieck's inequality will let us turn the first two players' entangled strategy into a classical one at a loss of an extra constant factor in the overall bias. 
We proceed with the formal proof of the theorem.

\begin{proof}[of Theorem~\ref{thm:questionbound}]
Suppose that the game $G$ is defined by the probability distribution $\pi$ and sign tensor $M$. Define the tensor $T_{ijk} = \pi(ijk)M(ijk)$.
By setting some entries to zero we may assume without loss of generality that  $T$ has dimension $Q\times Q\times Q$.
Fix an arbitrary constant $\epsilon>0$ and
let $\psiket$, $A_i,B_j,C_k$ be a finite-dimensional state and $\pmset{}$-valued observables such that\footnote{A standard approximation argument based on the Spectral Theorem shows that this is always possible.}
\beqn
\beta^*(G) \leq (1+\epsilon)\sum_{i,j,k=1}^Q T_{ijk}\psibra A_i\otimes B_j\otimes C_k\psiket.
\eeqn

Define for every $k\in[Q]$ the matrix $M_k = \sum_{i,j=1}^Q T_{ijk}A_i\otimes B_j$. Let $\eps_1,\dots,\eps_Q$ be i.i.d. $\pmset{}$-valued symmetrically distributed Bernoulli random variables. Using the fact that $\Exp[\eps_k\eps_\ell] = \delta_{k\ell}$ and the Cauchy-Schwarz inequality, the right-hand side of the above inequality can be written as and bounded by
\beqrn
\Exp\left[\Big(\psibra\sum_{k=1}^Q M_k\otimes (\eps_kI) \Big)\Big(\sum_{\ell=1}^Q\eps_{\ell} I\otimes C_\ell\psiket\Big)\right] &\leq& \Exp\left[\Big\|\psibra\sum_{k=1}^Q M_k\otimes (\eps_kI)\Big\|_2 \Big\|\sum_{\ell=1}^Q\eps_{\ell} I\otimes C_\ell\psiket\Big\|_2\right].
\eeqrn

Another application of Cauchy-Schwarz gives that the right-hand side is bounded from above by
\beq\label{eq:CStwo}
\left(\Exp\left[\Big\|\psibra\sum_{k=1}^Q M_k\otimes (\eps_kI)\Big\|_2^2\right]\right)^{1/2} \left(\Exp\left[\Big\|\sum_{\ell=1}^Q\eps_{\ell} I\otimes C_\ell\psiket\Big\|_2^2 \right]\right)^{1/2}.
\eeq

The fact that the matrices $\eps_\ell I\otimes C_\ell$ are unitary and $\psiket$ is a unit vector shows that the above term on the right equals $\sqrt{Q}$. Since the matrices $M_k\otimes (\eps_k I)$ are Hermitian, the left term in~\eqref{eq:CStwo} is at most
\beqn
\max_{\phiket,\,\zeta:[Q]\to\pmset{} }\phibra\sum_{k=1}^QM_k\otimes \big(\zeta(k)I\big)\phiket.
\eeqn

Expanding the definition of $M_k$, we have shown that
\beq\label{eq:3classical}
\beta^*(G) \leq (1+\epsilon)\sqrt{Q} \max_{\phiket,\,\zeta:[Q]\to\pmset{} }\phibra \sum_{i,j,k=1}^Q T_{ijk}A_i\otimes B_j\otimes\big(\zeta(k)I\big)\phiket.
\eeq
The matrices $\zeta(k)I$ may be interpreted as observables corresponding to single-outcome projective measurements. The outcome of such a measurement does not depend on the particular entangled state shared with the other players nor on their measurement outcomes.
The entangled bias of the game $G$ is thus at most $(1+\epsilon)\sqrt{Q}$ times the bias achievable with strategies in which the third player uses a classical  strategy.
The maximum on the right-hand side of \eqref{eq:3classical}  thus equals\footnote{Another way to see this  is by writing $ \sum_{i,j,k=1}^Q T_{ijk}A_i\otimes B_j\otimes\big(\zeta(k)I\big) =  \big(\sum_{i,j,k=1}^Q T_{ijk}A_i\otimes B_j\zeta(k)\big)\otimes I$ and using the facts that that the operator norm is multiplicative under tensor products and the identity matrix has operator norm~1.}
\beqn
 \max_{\ket{\phi'},\,\zeta:[Q]\to\pmset{} }\bra{\phi'} \sum_{i,j,k=1}^Q T_{ijk}A_i\otimes B_j\zeta(k)\ket{\phi'}.
 \eeqn

Let $\ket{\phi'}$ and $\zeta:[Q]\to\pmset{}$ be such that the maximum above is achieved. Define the $Q$-by-$Q$ matrix $H_{ij} = \sum_{k=1}^Q T_{ijk}\zeta(k)$. Rearranging terms gives that the above maximum equals $\sum_{i,j=1}^Q H_{ij} \bra{\phi'} A_i\otimes B_j\ket{\phi'}$. Define the unit vectors $x_i = A_i\otimes I\ket{\phi'}$ and $y_j = I\otimes B_j\ket{\phi'}$. Clearly we have $\bra{\phi'} A_i\otimes B_j\ket{\phi'} = \langle x_i,y_j\rangle$. The result now follows by applying Grothendieck's inequality~\eqref{eq:grothineq} and expanding the definition of $H_{ij}$.
\end{proof}

\subsection{Bounds in terms of the Hilbert space dimension}

In this section we give a proof of Theorem~\ref{thm:dimensionbound}, which we restate for convenience. 

\begin{theorem3} Let $G$ be a $3$-player XOR game in which the maximal entangled bias $\beta^*(G)$ is achieved by a strategy in which the third player's local dimension is $d$. Then 
$$\beta^*(G) \,\leq\, \sqrt{3d}\, \big(K_G^\C\big)^{3/2} \beta(G),$$
where $K_G^\C < 1.405$ is the complex Grothendieck constant. 
\end{theorem3}

As the bound in terms of the number of questions presented in the previous section, the proof of Theorem~\ref{thm:dimensionbound} relies on a decoupling technique, by which the third player is reduced to using a classical strategy, while only reducing the bias that the players achieve in the game by a factor depending on the local dimension of his share of the entangled state.
We use the following version of the non-commutative Khinchine's inequality, proved with optimal constants in~\cite{HM07}. 

\begin{theorem}[Khinchine's inequality, Proposition 2.12 in~\cite{HM07}]\label{thm:khinchine} Let $A_i$ be complex $d\times d$ matrices, and $\eps_i$ i.i.d. $\pmset{}$ symmetrically distributed. Then there exists a matrix random variable $\tilde{A}$ such that $\textsc{E}\big[ \eps_i \tilde{A}\big] = 0$ for every $i$, and for \emph{every} possible joint value taken by the tuple of random variables $(\eps_1,\ldots,\eps_d,\tilde{A})$ it holds that
\beq\label{eq:k-dual}
\Big\| \sum_i \eps_i A_i + \tilde{A} \Big\|_\infty \,\leq\, \sqrt{3} \max\Big\{ \Big\|\sum_i A_iA_i^\dagger \Big\|_\infty^{1/2},\,\Big\|\sum_i A_i^\dagger A_i \Big\|_\infty^{1/2}\Big\}.
\eeq
\end{theorem}

\medskip

\begin{proof}[of Theorem~\ref{thm:dimensionbound}]
Suppose that the game $G$ is defined by the probability distribution $\pi$ and sign tensor $M$. Define the tensor $T_{ijk} = \pi(ijk)M(ijk)$. Fix an arbitrary constant $\epsilon>0$ and
let $\psiket$, $A_i,B_j,C_k$ be a finite-dimensional state and $\pmset{}$-valued observables, where $C_k$ has dimension $d\times d$ and $A_i$, $B_j$ have (finite) dimension $D\times D$,  such that
\beqn
\beta^*(G) \leq (1+\epsilon)\sum_{i,j,k=1}^Q T_{ijk}\psibra A_i\otimes B_j\otimes C_k\psiket.
\eeqn
For each $k$, let $M_k = \sum_{i,j,k} T_{ijk}\, A_i\otimes B_j$.
 Let $\ket{\Psi} = \sum_i \lambda_i \ket{u_i}\ket{v_i}$ be the Schmidt decomposition, where $\ket{u_i}$ is a vector on the system held by the first two players, and $\ket{v_i}$ is on the third player's. Assume without loss of generality that the $\ket{v_i}$ span the local space of the third player. Letting $M = \sum_k M_k\otimes C_k$, the bias achieved by this strategy is $\bra{\Psi} M \ket{\Psi}\geq (1+\epsilon)^{-1}\beta^*(G)$. Decompose $M$ as $M = \sum_{i,j} E_{i,j} \otimes \ket{v_i}\bra{v_j}$, where for every $(i,j)\in [d]^2$ $E_{i,j}$ is a $D^2\times D^2$ matrix on Alice and Bob's systems; by definition 
$$ E_{i,j} = \sum_{k}  \bra{v_i} C_k \ket{v_j}\,M_k.$$
Since $M$ is Hermitian, we have $E_{i,j} = (E_{j,i})^\dagger$. We will need the following bound. 
\begin{claim}\label{claim:normbound} For every $i\in [d]$, 
\begin{align}
\max\Big\{\Big\| \sum_{j} E_{i,j} E_{i,j}^\dagger \big\|_\infty,\Big\| \sum_{j} E_{i,j}^\dagger E_{i,j}  \big\|_\infty\Big\} & \leq \big(K_G^\C\big)^3\, \beta(G)^2.\label{eq:mepsbound}
\end{align}
\end{claim}

\begin{proof}
Let $\ket{\Phi}$ be any vector. Then 
\begin{align}
\bra{\Phi}\sum_{j} E_{i,j}(E_{i,j})^\dagger \ket{\Phi} &= \sum_{k,k'} \sum_{j} \,\bra{\Phi} M_k M_{k'} \ket{\Phi} \bra{v_i} C_k \ket{v_j} \overline{\bra{v_i} C_{k'} \ket{v_j}}\notag\\
&= \sum_{k,k'} \, \bra{\Phi} M_k M_{k'} \ket{\Phi}\, \langle C_k^i,C_{k'}^i\rangle \notag \\
&\leq K_G^\C \, \sup_{c_k \in \{\pm 1\}} \sum_{k,k'} \, \bra{\Phi} M_k M_{k'} \ket{\Phi} \,c_kc_{k'}, \label{eq:mpsim2}
\end{align}
where in the second equality we let $C_k^i$ be the $i$-th row of $C_k$ (in the $\ket{v_j}$ basis), which has norm $1$ (since $C_k$ as a matrix is an observable), and the last inequality is Grothendieck's inequality. But
\begin{align*}
\sum_{k,k'} \, \bra{\Phi} M_k M_{k'} \ket{\Phi}\, c_kc_{k'} &\leq \Big\| \sum_{k,k'}\, M_k M_{k'} \,c_k c_{k'} \Big\|_\infty\\
& = \Big\| \sum_k \,M_k \,c_k\Big\|_\infty^2 \leq \big(K_G^\C\big)^2 \beta(G),
\end{align*}
again by Grothendieck's inequality: here the first two players are entangled players in $G$, but the third is classical. Using $E_{i,j} = E_{j,i}^\dagger$, this proves~\eqref{eq:mepsbound}. 
\end{proof}

Let $\eps_{j}$ be i.i.d. $\{\pm 1\}$-valued standard Bernoulli random variables, and for every $i\in [d]$ let $\tilde{E}_i$ the matrix random variable promised by Theorem~\ref{thm:khinchine}, and $E_i := \sum_{j} \eps_{j} E_{i,j} + \tilde{E}_i$. Combining the estimate in Claim~\ref{claim:normbound} with the bound~\eqref{eq:k-dual} from Theorem~\ref{thm:khinchine}, we get that 
\beq\label{eq:dim-0}
 \max_i \big\|E_i \big\|_\infty \,\leq\, \sqrt{3} \,\big(K_G^\C\big)^{3/2} \beta(G).
\eeq
Let $\eps'_i$ be i.i.d. $\{\pm 1\}$-valued standard Bernoulli random variables independent from the $\eps_i$  such that for all $(i,j)$, $\textsc{E}\big[\eps'_i \tilde{E}_{j}\big] = 0$. Starting from the $(E_i)$, let $\tilde{E}$ be the matrix random variable promised by Theorem~\ref{thm:khinchine}, and let $E := \sum_{i} \eps'_{i} E_{i} + \tilde{E}$. Using the triangle inequality, the bound~\eqref{eq:k-dual} together with~\eqref{eq:dim-0} leads to 
\beq\label{eq:dim-1}
\big\| E \big\|_\infty \,\leq\, \sqrt{3d}\, \big(K_G^\C\big)^{3/2} \beta(G),
\eeq
which is valid for all choices of $\eps_j$ and $\eps'_i$. We may now write
\begin{align} 
\bra{\Psi} M \ket{\Psi} &=  \sum_{i,j} \lambda_i \lambda_j \bra{u_i} E_{i,j} \ket{u_j}\notag \\
&= \text{E}_{\eps,\eps'}\Big[\Tr\Big( E \cdot \Big(\sum_{i,j}\eps'_{i}\eps_j \lambda_i\lambda_j \ket{u_j}\bra{u_i} \Big)\Big)\Big]\notag\\
&\leq \text{E}_{\eps,\eps'}\Big[\,\big\|E\big\|_\infty\,\Big\| \sum_{i,j} \eps'_{i}\eps_j \lambda_i\lambda_j \ket{u_j}\bra{u_i}\Big\|_1 \Big],\label{eq:mpsim}
\end{align}
where for the second equality we used that $\textsc{E}\big[ \eps'_{i} \tilde{E}\big]=0$ for every $i$, and the last follows from H\"older's inequality. The norm $\|E\|_\infty$ is bounded by~\eqref{eq:dim-1}, and to conclude it suffices to note that, since 
$$\sum_{i,j} \eps'_{i}\eps_j \lambda_i\lambda_j \ket{u_j}\bra{u_i}\,=\, \Big(\sum_j \eps_j \lambda_j \ket{u_j}\Big)\Big( \sum_i \eps'_i \lambda_i \bra{u_i}\Big),$$
its trace norm is at most
$$ \Big\|\sum_j \eps_j \lambda_j \ket{u_j}\Big\|\Big\| \sum_i \eps'_i \lambda_i \ket{u_i}\Big\| \,\leq\, \Big(\sum_j \lambda_j^2\Big)^{1/2}\Big(\sum_i \lambda_i^2 \Big)^{1/2}\,\leq\, 1.$$
\end{proof}

\section{Conclusion and open problems}\label{sec:conclusion}

We have described a probabilistic construction of a family of XOR games $G = (G_N)$ in which players sharing entanglement may gain a large, unbounded advantage over the best classical, unentangled players. For any $N=2^n$ the game $G_N$ has $N^2$ questions per player, and is such that the ratio
$\beta^*(G)/\beta(G) = \Omega\big( \sqrt{N} \log^{-5/2} N)$. Our results raise two immediate open questions. The first is whether this estimate is optimal: we could only prove an upper bound of $O(N)$ on the largest possible ratio (for games, such as $G_N$, with at most $N^2$ questions per player). The second is to give an explicit, deterministic construction of a family of games achieving a similar (or even weaker) ratio. Such a construction would be of great interest both to experimental physicists and to operator space theorists, no small feat! 

In our results we measured the advantage of entangled players in a given XOR game $G$ multiplicatively, as a function of the ratio $\beta^*(G)/\beta(G)$. Although this has become customary, if one is interested in experimental realizations it may not be the most appropriate way to measure the advantage gained by entanglement, as small biases may be hard to notice, however large the ratio between the entangled and unentangled biases. In the case of our specific construction, one may compute that $\beta^*(G_N) = \Omega(N^{-3/2})$ and $\beta(G_N) = O(N^{-2}\log^{5/2} N)$: while the ratio of these two quantities is large, both are relatively close to $0$ and may thus be difficult to differentiate through experiment. It is an interesting open problem to also obtain large separations as measured, say, by the \emph{difference} $\beta^*(G)-\beta(G)$.

\paragraph{Acknowledgements.} We thank Harry Buhrman, Carlos Palazuelos, David Perez-Garcia, Gilles Pisier, Oded Regev, Ignacio Villanueva and Ronald de Wolf for useful discussions. T.V. is grateful to CWI for hosting him while part of this work was done. J.B. thanks Universidad Complutense for hosting him while part of this work was done. 

\bibliography{newref}
\end{document}